\input ./style/arxiv-vmsta.cfg
\documentclass[numbers,compress,v1.0.1]{vmsta}
\usepackage{vtexbibtags}
\usepackage{mathbh}
\usepackage{authorquery}

\volume{5}
\issue{3}
\pubyear{2018}
\firstpage{317}
\lastpage{336}
\aid{VMSTA107}
\doi{10.15559/18-VMSTA107}
\articletype{research-article}

\startlocaldefs

\newcommand{\rrvert}{\vert}
\newcommand{\llvert}{\vert}
\allowdisplaybreaks
\numberwithin{equation}{section}

\newtheorem{thm}{Theorem}[section]
\newtheorem{cor}[thm]{Corollary}

\newtheorem{prop}[thm]{Proposition}
\theoremstyle{definition} 

\newtheorem{exmp}[thm]{Example}
\newtheorem{rem}[thm]{Remark}

\endlocaldefs

\begin{aqf}
\querytext{Q1}{Isn't it a 'damping parameter'?}
\querytext{Q2}{Isn't it a 'damping parameter'?}
\end{aqf}
\begin{document}

\begin{frontmatter}
\pretitle{Research Article}

\title{Cliquet option pricing in a jump-diffusion L\'{e}vy model}

\author{\inits{M.}\fnms{Markus}~\snm{Hess}\ead[label=e1]{Markus-Hess@gmx.net}} 
\address{
Independent}



\markboth{M. Hess}{Cliquet option pricing in a jump-diffusion L\'{e}vy model}

\begin{abstract}
We investigate the pricing of cliquet options in a
jump-diffusion model. The considered option is of monthly sum cap style
while the underlying stock price model is driven by a drifted L\'{e}vy
process entailing a Brownian diffusion component as well as compound
Poisson jumps. We also derive representations for the density and
distribution function of the emerging L\'{e}vy process. In this setting, we
infer semi-analytic expressions for the cliquet option price by two
different approaches. The first one involves the probability distribution
function of the driving L\'{e}vy process whereas the second draws upon
Fourier transform techniques. With view on sensitivity analysis and hedging
purposes, we eventually deduce representations for several Greeks while
putting emphasis on the Vega.
\end{abstract}

\received{\sday{5} \smonth{4} \syear{2018}}
\revised{\sday{14} \smonth{6} \syear{2018}}
\accepted{\sday{23} \smonth{6} \syear{2018}}

\publishedonline{\sday{20} \smonth{7} \syear{2018}\phantom{\quad}}
\end{frontmatter}

\section{Introduction}\label{sec1}

During the last decades, cliquet option based contracts became a very
popular and frequently sold investment product in the insurance industry.
These contracts can be considered as a customized subclass of equity
indexed annuities which combine savings and insurance benefits \cite{3}. The
underlying options usually are of monthly sum cap style paying a credited
yield based on the sum of monthly-capped rates associated with some
reference (stock) index. More precisely, the investor pays a contractually
specified amount to the issuer of the option prior to its maturity date
and, in turn, receives at maturity a payoff depending on the performance of
some designated reference index. In this regard, cliquet type investments
belong to the class of path-dependent exotic options. The most popular
choice in the insurance branch are cliquet contracts with globally-floored
and locally-capped payoffs. These products can be utilized to protect
against downside risk while yielding significant upside potential, yet
avoiding extreme payoffs due to their local capping (cf. \cite{3,10,16}).
In \cite{16} cliquet options are regarded as ``the height of fashion in the
world of equity derivatives''.\looseness=-1

In the literature, there are different pricing approaches for cliquet
options involving e.g. partial differential equations (see \cite{16}), Monte
Carlo techniques (see \cite{1,2,5,9}), numerical recursive algorithms
related to inverse Laplace transforms (see \cite{10}) and analytical computation
methods (see \cite{1,3,5,8,9,11}). In \cite{3} the authors provide
semi-analytic pricing formulas for path-dependent equity-linked contracts.
They distinguish between cliquet options and monthly sum cap contracts and
derive expressions for various Greeks. In their approach, it is crucial to
know the probability distribution of the returns of the underlying
reference index. In \cite{1} the author uses a L\'{e}vy process specification to
model the evolution of the underlying reference portfolio and investigates
the valuation of life insurance policies providing interest rate
guarantees. The driving L\'{e}vy process is of jump-diffusion type with
normally distributed jump amplitudes while a special focus in \cite{1} is laid
on valuation under different risk-neutral pricing measures. In \cite{11} the
valuation of insurance contracts is discussed while emphasis is put on the
impact of different L\'{e}vy process model specifications. It is shown that
changing the underlying asset model implies a significant change in the
prices of guarantees, indicating a substantial model risk. In \cite{8} the
pricing of cliquet options in a geometric Meixner model is investigated.
The considered option is of monthly sum cap style while the underlying
stock price is driven by a pure-jump Meixner--L\'{e}vy process yielding
Meixner distributed log-returns. The paper \cite{8} provides a specific
application of the results derived in the present article. In \cite{10} cliquet
option prices are computed numerically by a recursive algorithm involving
inverse Laplace transforms. This method is applied to a lognormal and a
jump-diffusion model with deterministic volatility as well as to the Heston
stochastic volatility model. In addition, a sensitivity analysis in each
model is presented. Moreover, in \cite{7} cliquet option pricing in a
jump-diffusion model with time-dependent coefficients is examined. The
jumps in the stock price trajectory are interspersed by an increasing
standard Poisson process and the time-dependent coefficients are
approximated by piecewise constant functions. In \cite{7} there are solely
cliquet options with a single resetting time discussed. In \cite{16} the author
investigates cliquet option pricing with partial differential equations
(PDEs) while putting a special focus on the important role of volatility
surface modeling. Recently, there have been extensions beyond L\'{e}vy
settings to regime switching L\'{e}vy models (see e.g. \cite{5,9}). In \cite{5}
the authors investigate the pricing of equity-linked annuities with
cliquet-style guarantees in regime-switching stochastic volatility models
with jumps. They propose a transform-based pricing method involving density
projections and continuous-time Markov chain approximations. The considered
models include exponential and regime-switching L\'{e}vy processes as well
as stochastic volatility models with general jump size distributions. In
\cite{9} the valuation of equity-linked life insurance contracts in a regime
switching L\'{e}vy model is studied. The model parameters depend on a
continuous-time finite-state Markov chain, and closed-form pricing formulas
based on Fourier transform techniques are derived.\looseness=-1

The aim of the present paper is to provide analytical pricing formulas for
globally-floored locally-capped cliquet options with multiple resetting
times where the underlying reference stock index is driven by a drifted
time-homogeneous L\'{e}vy process with Brownian diffusion component and
compound Poisson jumps. In our framework, jumps represent rare events such
as crashes, large drawdowns or upward movements. The dates of e.g. market
crashes are modeled as arrival times of a standard Poisson process while
the jump amplitudes can be both positive and negative. With reference to
Section 4.1.1 in \cite{4}, we state that jump-diffusion models are easy to
simulate and efficient Monte Carlo methods for option pricing are
available. Jump-diffusion models also perform well when it comes to implied
volatility smile interpolation (see Section 13 in \cite{4}). In our setup, we
derive cliquet option price formulas under two different approaches: once
by using the distribution function of the driving L\'{e}vy process and once
by applying Fourier transform techniques. In the context of sensitivity
analysis, we eventually provide expressions for several Greeks related to
our model.

The paper is organized as follows: In Section~\ref{sec2} we introduce the
jump-diffusion stock price model and derive representations for the
probability density and distribution function of the driving L\'{e}vy
process. In Section~\ref{sec3} we are concerned with cliquet option pricing under
both a distribution function and a Fourier transform approach. Section~\ref{sec4} is
dedicated to sensitivity analysis and the computation of different Greeks.
In Section~\ref{sec5} we draw the conclusions and briefly mention some future
research topics.

\section{A L\'{e}vy stock price model and its distributional properties}\label{sec2}

Let $ ( \Omega, \mathbb{F},  ( \mathcal{F}_{t}  )_{t\in  [
0, T  ]}, \mathbb{Q}  )$ be a filtered probability space
satisfying the usual hypotheses, i.e. $\mathcal{F}_{t} = \mathcal{F}_{t +}
:= \cap_{s > t} \mathcal{F}_{s}$ constitutes a right-continuous
filtration and $\mathbb{F}$ denotes the sigma-algebra augmented by all
$\mathbb{Q}$-null sets (cf. p. 3 in \cite{13}). Here, $\mathbb{Q}$ is a
risk-neutral probability measure and $0< T < \infty$ denotes a finite time
horizon. In the sequel, we introduce a stochastic model for the stock price
process $S_{t}$. Let $t\in  [ 0, T  ]$ and consider the stochastic
differential equation (SDE)
\begin{align*}
d S_{t} = \eta ( t, S_{t} ) \,dt + \sigma ( t, S_{t}
) \,d W_{t} + \int_{\mathbb{R}} \theta ( t, z,
S_{t-} ) \,dN ( t, z )
\end{align*}
where $\eta$, $\sigma$ and $\theta$ are deterministic functions, $W$
constitutes an $\mathcal{F}$-adapted standard Brownian motion under
$\mathbb{Q}$ and $N$ is a Poisson random measure (PRM). We further
introduce the $\mathbb{Q}$-compensated PRM
%
\begin{align}
\label{eq2.1} d \widetilde{N} ( s, z ) :=dN ( s, z ) -d\nu ( z ) \,ds
\end{align}
which constitutes an $ ( \mathcal{F}, \mathbb{Q}  )$-martingale
integrator on $ [ 0, T  ] \times \mathbb{R}$ with positive and
finite L\'{e}vy measure $\nu$ satisfying $\nu  (  \{ 0  \}
 ) =0$ and
\begin{align*}
\int_{\mathbb{R}} \bigl( 1 \wedge z^{2} \bigr) \,d\nu ( z )
< \infty
\end{align*}
(cf. Eq. (3.14) in \cite{4}). In the above setup, we refer to $\eta$, $\sigma$
and $\theta$ as the drift, volatility and jump function, respectively. We
assume that $W$ and $N$ are $\mathbb{Q}$-independent and set
\begin{align*}
\mathcal{F}_{t} :=\sigma \{ S_{u}:0 \leq u \leq t \}
\end{align*}
for all $t\in  [ 0, T  ]$. In the next step, we specify the
emerging coefficients as follows
\begin{align*}
\eta ( t, S_{t} ) :=\eta ( t ) S_{t}, \qquad \sigma ( t,
S_{t} ) :=\sigma ( t ) S_{t}, \qquad \theta ( t, z, S_{t-}
) :=\theta ( t, z ) S_{t-}
\end{align*}
while assuming that $\theta  ( t, z  ) > - 1$ for all $ ( t, z
 ) \in  [ 0, T  ] \times \mathbb{R}$ and that
\begin{align*}
\mathbb{E}_{\mathbb{Q}} \Biggl[ \int_{0}^{T}
\biggl( \bigl\llvert \eta ( t ) \bigr\rrvert + \sigma^{2} ( t ) + \int
_{\mathbb{R}} \theta^{2} ( t, z ) \,d\nu ( z ) \biggr) \,dt
\Biggr] < \infty
\end{align*}
(cf. Section 9.1 in \cite{6}). Consequently, we obtain the geometric SDE
\begin{align*}
\frac{d S_{t}} {S_{t-}} = \eta ( t ) \,dt + \sigma ( t ) \,d W_{t} + \int
_{\mathbb{R}} \theta ( t, z ) \,dN ( t, z )
\end{align*}
which possesses the discontinuous Dol\'{e}ans-Dade solution
\begin{align*}
S_{t} &= S_{0} \exp \Biggl\{ \int_{0}^{t}
\biggl( \eta ( s ) - \frac{1} {2} \sigma^{2} ( s ) \biggr) \,ds +
\int_{0}^{t} \sigma ( s ) \,d W_{s}
\\
&\quad + \int_{0}^{t} \int_{\mathbb{R}}
\ln \bigl( 1+ \theta ( s, z ) \bigr) \,dN ( s, z ) \Biggr\}
\end{align*}
for all $t\in  [ 0, T  ]$. From now on, we set $\eta  ( t
 ) \equiv \eta$, $\sigma  ( t  ) \equiv \sigma >0$ and
$\theta  ( t, z  ) := e^{z} - 1$ in order to obtain a
time-homogeneous L\'{e}vy process specification. If we do so, the latter
equation can be written as
%
\begin{align}
\label{eq2.2} S_{t} = S_{0} e^{X_{t}}
\end{align}
with a real-valued L\'{e}vy process
%
\begin{align}
\label{eq2.3} X_{t} :=\gamma t + \sigma W_{t} + \int
_{0}^{t} \int_{\mathbb{R}} z \,dN ( s, z )
\end{align}
where $\gamma:=\eta- {\sigma^{2}} / {2}$ and $t\in  [ 0, T
 ]$. Note that $X_{0} =0$ $\mathbb{Q}$-a.s. We denote the characteristic
triplet of $X$ by $ ( \gamma, \sigma, \nu  )$. Moreover, the first
moment of $X$ is given by
\begin{align*}
\mathbb{E}_{\mathbb{Q}} [ X_{t} ] = t \biggl( \gamma + \int
_{\mathbb{R}} z \,d\nu ( z ) \biggr)
\end{align*}
whereas the characteristic function of $X$ can be computed by the
L\'{e}vy--Khinchin formula (see e.g. \cite{4,6,14,15}) due to
%
\begin{align}
\label{eq2.4} \phi_{X_{t}} ( u ) := \mathbb{E}_{\mathbb{Q}} \bigl[
e^{iu X_{t}} \bigr] = e^{\psi  ( u  ) t}
\end{align}
with $i^{2} = - 1$, $u \in \mathbb{R}$, $t\in  [ 0, T  ]$ and
characteristic exponent
%
\begin{align}
\label{eq2.5} \psi ( u ) :=iu\gamma- \frac{1} {2} \sigma^{2}
u^{2} + \int_{\mathbb{R}} \bigl[ e^{iuz} - 1 \bigr]
\,d\nu ( z ).
\end{align}
More details on L\'{e}vy processes can be found in e.g. \cite{4,6,14,15}. In the next step, we define the discounted stock price
\begin{align*}
\hat{S}_{t} := \frac{S_{t}} {B_{t}}\vadjust{\goodbreak}
\end{align*}
where $S_{t}$ is such as defined in (\ref{eq2.2}) and $B_{t} := e^{rt}$ is
the value of a bank account with normalized initial capital $B_{0} =1$ and
risk-less interest rate $r>0$. Due to (\ref{eq2.2}), we find
\begin{align*}
\hat{S}_{t} = S_{0} e^{X_{t} - rt}
\end{align*}
while It\^{o}'s formula yields the following geometric SDE under
$\mathbb{Q}$
\begin{align*}
\frac{d \hat{S}_{t}} {\hat{S}_{t-}} = \biggl( \eta-r + \int_{\mathbb{R}} \bigl[
e^{z} - 1 \bigr] \,d\nu ( z ) \biggr) \,dt + \sigma \,d W_{t} +
\int_{\mathbb{R}} \bigl[ e^{z} - 1 \bigr] \,d \widetilde{N} ( t, z
).
\end{align*}
In accordance to no-arbitrage theory, the discounted stock price process
$\hat{S}_{t}$ must form a martingale under the risk-neutral probability
measure $\mathbb{Q}$. For this reason, we have to require the drift
restriction
\begin{align*}
\eta = r- \int_{\mathbb{R}} \bigl[ e^{z} - 1 \bigr] \,d\nu (
z ).
\end{align*}
With this particular choice of the drift coefficient $\eta$, we obtain
\begin{align*}
\frac{d S_{t}} {S_{t-}} = r \,dt + \sigma \,d W_{t} + \int_{\mathbb{R}}
\bigl[ e^{z} - 1 \bigr] \,d \widetilde{N} ( t, z )
\end{align*}
under $\mathbb{Q}$. Summing up, if we model the stock price process $S_{t}$
as in the latter equation, then the discounted stock price $\hat{S}_{t}$
constitutes a $\mathbb{Q}$-martingale.

Furthermore, let us define the Fourier transform, respectively inverse
Fourier transform, of a function $q\in \mathcal{L}^{1}  ( \mathbb{R}
 )$ via
\begin{align*}
\hat{q} ( y ) := \int_{\mathbb{R}} q ( x ) e^{iyx} \,dx,
\qquad q ( x ) = \frac{1} {2 \pi} \int_{\mathbb{R}} \hat{q} ( y )
e^{-iyx} \,dy.
\end{align*}

\begin{prop}[density function]
\label{prop2.1}
{Suppose that the
L\'{e}vy process} $X_{t}$ {is such as defined in} (\ref{eq2.3}).
{Then for all} $t\in  [ 0, T  ]$ {and} $x
\in \mathbb{R}$ {the probability density function} $f_{X_{t}}  (
x  )$ of $X_{t}$ {under} $\mathbb{Q}$ {can be represented
as}
%
\begin{align}
\label{eq2.6} f_{X_{t}} ( x ) = \frac{1} {2 \pi} \int_{\mathbb{R}}
\exp \biggl\{ -iux + t \biggl( iu\gamma- \frac{1} {2} \sigma^{2}
u^{2} + \int_{\mathbb{R}} \bigl[ e^{iuz} - 1 \bigr]
\,d\nu ( z ) \biggr) \biggr\} \,du.
\end{align}
\end{prop}

\begin{proof} Note that the characteristic function (\ref{eq2.4}) is the Fourier
transform of the density function $f_{X_{t}}  ( \boldsymbol{\cdot}
 )$, that is,
\begin{align*}
\phi_{X_{t}} ( u ) = \int_{\mathbb{R}} e^{iux}
f_{X_{t}} ( x ) \,dx.
\end{align*}
We next apply the inverse Fourier transform and hereafter take (\ref{eq2.4}) and
(\ref{eq2.5}) into account which yields the density function (\ref{eq2.6}). \end{proof}

In what follows, we investigate in detail the jump part of the L\'{e}vy
process $X$ denoted by
\begin{align*}
L_{t} := \int_{0}^{t} \int
_{\mathbb{R}} z \,dN ( s, z ) = \sum_{j =1}^{N_{t}}
Y_{j}
\end{align*}
which constitutes a c\`{a}dl\`{a}g, finite activity compound Poisson
process (CPP) with finitely many jumps in each time interval. In the latter
equation, $N_{t}$ constitutes a standard Poisson process under $\mathbb{Q}$
with deterministic jump intensity $\lambda >0$, i.e. $N_{t} \sim Poi  (
\lambda t  )$, while $Y_{1}, Y_{2}, \dots$ are i.i.d. random variables
modeling the jump amplitudes. We put $\beta:=
\mathbb{E}_{\mathbb{Q}}  [ Y_{1}  ]$. Recall that the compensated
compound Poisson process $ ( L_{t} -\beta\lambda t  )_{t}$
constitutes an $ ( \mathcal{F}_{t}, \mathbb{Q}  )$-martingale
which implies
\begin{align*}
\beta\lambda = \int_{\mathbb{R}} z \,d\nu ( z )
\end{align*}
thanks to (\ref{eq2.1}). Note that $N_{t}$ shall not be mixed up with the Poisson
random measure $dN  ( s, z  )$. Obviously, we may write $X_{t} =
\gamma t + \sigma W_{t} + L_{t}$. We assume that the stochastic processes
$W_{t}$, $N_{t}$ and the random variables $Y_{1}, Y_{2}, \dots$ altogether
are $\mathbb{Q}$-independent.

\begin{exmp}\label{exmp2.2} If $Y_{j}$ is normally distributed with mean $\mu$
and variance $\delta^{2}$ under $\mathbb{Q}$ for all $j$, then the L\'{e}vy
measure possesses the Lebesgue density $d\nu  ( z  ) = \lambda
\varphi_{\mu, \delta^{2}}  ( z  ) \,dz$ where
\begin{align*}
\varphi_{\mu, \delta^{2}} ( z ) := \frac{1} {\sqrt{2
\pi \delta^{2}}} e^{- \frac{1} {2}  ( \frac{z-\mu} {\delta}
 )^{2}}
\end{align*}
and $z \in \mathbb{R}$. Here, $\mu$ and $\delta^{2}$ model the mean
respectively variance of the jump sizes. In this setup, we receive $\beta =
\mu$, $\mathbb{E}_{\mathbb{Q}}  [ X_{t}  ] = t  ( \gamma +
\lambda\mu  )$ and $\mathbh{Var}_{\mathbb{Q}}  [ X_{t}  ] = t
( \sigma^{2} + \lambda \delta^{2} + \lambda \mu^{2} )$.
Evidently, choosing a negative $\mu$ makes the occurrence of downward jumps
more likely than upward jumps and vice versa. We remark that a similar
model specification with normally distributed jump sizes has firstly been
proposed in \cite{12}.
\end{exmp}

\begin{exmp}\label{exmp2.3} If $Y_{j}$ is exponentially distributed with
parameter $\alpha>0$ under $\mathbb{Q}$ for all~$j$, then the L\'{e}vy
measure possesses the Lebesgue density $d\nu  ( z  ) =\lambda
p_{\alpha}  ( z  ) \, dz$ where $p_{\alpha}  ( z  )
:= \alpha e^{- \alpha z}$ and $z \in  [ 0, \infty  [$.
We presently find $\beta= {1} / {\alpha}$.
\end{exmp}

\begin{cor}\label{cor2.4}
\begin{enumerate}
\item[(a)] {Suppose that} $Y_{j}$ {is
normally distributed} (cf.~\cite{12}) {with mean} $\mu$ {and
variance} $\delta^{2}$ {under} $\mathbb{Q}$ {for all} $j$.
{Then for all} $t \in  [ 0,T  ]$ {and} $x
\in \mathbb{R}$ {the probability density function of} $X_{t}$
{under} $\mathbb{Q}$ {takes the form}
\begin{align*}
f_{X_{t}} ( x ) = \frac{1} {2\pi} \int_{\mathbb{R}} \exp
\biggl\{ - iux+t \biggl( iu\gamma - \frac{1} {2} \sigma^{2}
u^{2} +\lambda e^{iu\mu -
{\delta^{2} u^{2}} / {2}} - \lambda \biggr) \biggr\} \,du.
\end{align*}
\item[(b)] {Suppose that} $Y_{j}$ {is exponentially distributed
with parameter} $\alpha>0$ {under} $\mathbb{Q}$ {for all}
$j$. {Then for all} $t \in  [ 0,T  ]$ {and} $x
\in \mathbb{R}$ {the probability density function of} $X_{t}$
{under} $\mathbb{Q}$ {takes the form}
\begin{align*}
f_{X_{t}} ( x ) = \frac{1} {2\pi} \int_{\mathbb{R}} \exp
\biggl\{ i ( \gamma t - x ) u - \frac{1} {2} t \sigma^{2}
u^{2} - \frac{\lambda tu} {u+i\alpha} \biggr\} \,du.
\end{align*}
\end{enumerate}
\end{cor}

\begin{proof} Combine (\ref{eq2.6}) with Example~\ref{exmp2.2} and Example~\ref{exmp2.3}.
\end{proof}

\begin{prop}[distribution function]
\label{prop2.5}
 {Let} $X_{t}
=\gamma t+\sigma W_{t} + \sum_{j=1}^{N_{t}} Y_{j}$ {and assume that
the standard Poisson process} $N_{t}$ {jumps} $m$ {times in
the time interval} $ [ 0,t  ]$, {that is}, $N_{t} =m$
{with} $m \in \mathbb{N}_{0}$. {As in}~\cite{12},
{suppose that} $Y_{j}$ {is normally distributed with mean}
$\mu$ {and variance} $\delta^{2}$. {Then for any Borel set}
$A \subset \mathbb{R}$ {and} $t \in  [ 0,T  ]$ {the
cumulative probability distribution function of} $X_{t}$ {under}
$\mathbb{Q}$ {possesses the representation}
%
\begin{align}
\label{eq2.7} \mathbb{Q} ( X_{t} \in A ) = \int_{A}
e^{- \lambda t} \sum_{m=0}^{\infty}
\frac{ ( \lambda t  )^{m}} {m! \sqrt{2\pi
 ( \sigma^{2} t+m \delta^{2}  )}} \exp \biggl\{ - \frac{1} {2}
 \frac{ ( x - \gamma t - m\mu  )^{2}} {\sigma^{2} t+m \delta^{2}} \biggr\}
\,dx.
\end{align}
\end{prop}

\begin{proof} Let $A \subset \mathbb{R}$ and $t \in  [ 0,T  ]$.
In accordance to Section 4.3 in \cite{4} and the properties of conditional
probabilities, the following (``quickly converging'' \cite{4}) series
representation for the distribution function of $X_{t}$ under $\mathbb{Q}$
holds
%
\begin{align}
\label{eq2.8} \mathbb{Q} ( X_{t} \in A ) &= \sum
_{m=0}^{\infty} \mathbb{Q} \bigl( \{ X_{t} \in A
\} \cap \{ N_{t} =m \} \bigr)
\nonumber
\\
&= \sum_{m=0}^{\infty} \mathbb{Q} (
X_{t} \in A \mid N_{t} =m ) \mathbb{Q} ( N_{t} =m )
\nonumber
\\
& = e^{- \lambda t} \sum_{m=0}^{\infty}
\mathbb{Q} ( X_{t} \in A \mid N_{t} =m ) \frac{ ( \lambda t
 )^{m}} {m!}
\end{align}
wherein
\begin{align*}
\mathbb{Q} ( X_{t} \in A \mid N_{t} =m ) = \mathbb{Q}
\Biggl( \Biggl( \gamma t+\sigma W_{t} + \sum
_{j=1}^{m} Y_{j} \Biggr) \in A \Biggr).
\end{align*}
Since $Y_{j} \sim \mathcal{N}  ( \mu, \delta^{2}  )$ for all $j$,
we find that the stochastic process $( \gamma t+\sigma W_{t} +
\sum_{j=1}^{m} Y_{j} )_{t}$ also is normally distributed under
$\mathbb{Q}$ with mean $\gamma t+m\mu$ and variance $\sigma^{2} t+m
\delta^{2}$. Thus, by the definition of the cumulative distribution
function, we get
\begin{align*}
\mathbb{Q} \Biggl( \Biggl( \gamma t+\sigma W_{t} + \sum
_{j=1}^{m} Y_{j} \Biggr) \in A \Biggr) =
\int_{A} \varphi_{\gamma t+m\mu, \sigma^{2} t+m
\delta^{2}} ( x ) \,dx
\end{align*}
where $\varphi$ denotes the probability density function of the normal
distribution (see \cite{12} and Example~\ref{exmp2.2} above). Putting the latter equations
together, we end up with (\ref{eq2.7}). \end{proof}

\begin{rem}\label{rem2.6} Verify that the proof of Proposition~\ref{prop2.5} only works,
if the random variables $Y_{j}$ are normally distributed for every $j$. If
$Y_{j}$ is e.g. exponentially distributed for all $j=1,\dots,m$ (as
proposed in Example~\ref{exmp2.3}), then it is unclear how to compute the
probability
\begin{align*}
\mathbb{Q} \Biggl( \Biggl( \gamma t+\sigma W_{t} + \sum
_{j=1}^{m} Y_{j} \Biggr) \in A \Biggr)
\end{align*}
emerging in the sequel of (\ref{eq2.8}).
\end{rem}

\begin{cor}[Eq. (4.12) in \cite{4}]\label{cor2.7}  Under the assumptions
of Proposition~\ref{prop2.5}, for all $t\in  [ 0,T  ]$
{and} $x \in \mathbb{R}$ {the probability density function
of} $X_{t}$ {under} $\mathbb{Q}$ {is given by}
%
\begin{align}
\label{eq2.9} f_{X_{t}} ( x ) = e^{-\lambda t} \sum
_{m=0}^{\infty} \frac{ ( \lambda t  )^{m}} {m! \sqrt{2\pi  ( \sigma^{2} t+m
\delta^{2}  )}} \exp \biggl\{ -
\frac{1} {2} \frac{ ( x-\gamma
t-m\mu  )^{2}} {\sigma^{2} t+m \delta^{2}} \biggr\}.
\end{align}
\end{cor}

\begin{proof} The density can directly be read off in (\ref{eq2.7}). Also see
\cite{12}. \end{proof}

\begin{cor}\label{cor2.8} {If the Borel set} $A=  \,] -\infty,a
 ] \subseteq \mathbb{R}$ {is an interval, then for any} $a
\in \mathbb{R}$ {and} $t\in  [ 0,T  ]$ {the
distribution function in} (\ref{eq2.7}) {takes the form}
\begin{align*}
\mathbb{Q} ( X_{t} \leq a ) = e^{-\lambda t} \sum
_{m=0}^{\infty} \frac{ ( \lambda t  )^{m}} {m!} \varPhi \biggl(
\frac{a-\gamma t-m\mu}{
\sqrt{\sigma^{2} t+m \delta^{2}}} \biggr)
\end{align*}
{where} $\varPhi$ {denotes the standard normal cumulative
distribution function}.
\end{cor}

\begin{proof} This representation is an immediate consequence of
Proposition~\ref{prop2.5}. \end{proof}

Recall that the stochastic process $S_{t}$ will serve as our stock price
model when it comes to cliquet option pricing in the subsequent section. In
this context, for $n \in \mathbb{N}$ we introduce the time partition
$\mathcal{P}:=  \{ 0< t_{0} < t_{1} <\cdots< t_{n} \leq T
 \}$ and define the return/revenue process associated with the period
$ [ t_{k-1}, t_{k}  ]$ via
%
\begin{align}
\label{eq2.10} R_{k} := \frac{S_{t_{k}} - S_{t_{k-1}}} {S_{t_{k-1}}} = e^{X_{t_{k}} - X_{t_{k-1}}} -1
\end{align}
where $k\in  \{ 1,\dots,n  \}$ and $X$ is the L\'{e}vy process
defined in (\ref{eq2.3}). Note that $R_{1},\dots, R_{n}$ are
$\mathbb{Q}$-independent and that $R_{k} >-1$ $\mathbb{Q}$-almost sure for
all $k$. For the sake of notational simplicity, we always work under the
assumption of equidistant time points in the following and define
$\tau:= t_{k} - t_{k-1}$. If we want to refrain from this
assumption again, then $\tau$ simply has to be replaced by the difference
$t_{k} - t_{k-1}$ in all subsequent equations -- with (\ref{eq3.17}) as an
exception.

\begin{prop}\label{prop2.9}
{Let} $\mathcal{P}=  \{ 0< t_{0} <
t_{1} <\cdots< t_{n} \leq T  \}$ {and put} $\tau= t_{k} -
t_{k-1}$ {for} $k\in  \{ 1,\dots,n  \}$
({equidistant time points}). {Define the return process}
$R_{k}$ {as in} (\ref{eq2.10}). {Then for any fixed real-valued}
$\xi>-1$ {the distribution function of} $R_{k}$ {under}
$\mathbb{Q}$ {admits the series representation}
%
\begin{align}
\label{eq2.11} \mathbb{Q} ( R_{k} \leq\xi ) = e^{-\lambda\tau} \sum
_{m=0}^{\infty} \frac{ ( \lambda\tau  )^{m}} {m!} \varPhi
\biggl( \frac{\ln  ( 1+\xi  ) -\gamma\tau-m\mu} {\sqrt{\sigma^{2} \tau+m
\delta^{2}}} \biggr)
\end{align}
{where} $\varPhi$ {denotes the standard normal cumulative
distribution function}.
\end{prop}

\begin{proof} Since $X$ is a L\'{e}vy process under $\mathbb{Q}$, we
observe $X_{t_{k}} - X_{t_{k-1}} \cong X_{\tau}$ (stationary increments)
where $\tau= t_{k} - t_{k-1}$ and the symbol $\cong$ denotes equality in
distribution. Taking (\ref{eq2.10}) and (\ref{eq2.9}) into account, we obtain
%
\begin{align}
\label{eq2.12} \mathbb{Q} ( R_{k} \leq\xi ) &= \mathbb{Q} \bigl(
X_{\tau} \leq \ln ( 1+\xi ) \bigr)
\nonumber
\\
&= e^{-\lambda\tau} \sum_{m=0}^{\infty}
\frac{ ( \lambda\tau  )^{m}} {m!} \int_{-\infty}^{\ln  ( 1+\xi
 )}
\varphi_{\gamma\tau+m\mu, \sigma^{2} \tau+m \delta^{2}} ( x ) \,dx
\end{align}
with
%
\begin{align}
\label{eq2.13} f_{X_{\tau}} ( x ) = e^{-\lambda\tau} \sum
_{m=0}^{\infty} \frac{ ( \lambda\tau  )^{m}} {m!} \varphi_{\gamma\tau+m\mu,
\sigma^{2} \tau+m \delta^{2}}
( x )
\end{align}
where $\varphi$ denotes the density function of the normal distribution
(recall Example~\ref{exmp2.2}). We finally perform the integration and end up with
(\ref{eq2.11}). \end{proof}

In quantitative risk management, it is often of interest to compute the
probability of large drawdowns (shocks) in asset prices like e.g.
$\mathbb{Q}  ( S_{u} \leq\kappa S_{t}  )$, $0\leq t<u\leq T$,
where $\kappa$ constitutes some stress scenario percentile like 60\%, for
instance. Due to (\ref{eq2.2}), we find $\mathbb{Q}  ( S_{u} \leq\kappa S_{t}
 ) = \mathbb{Q}  ( X_{u-t} \leq \ln \kappa  )$ which can
easily be computed by Corollary~\ref{cor2.8}.

\section{Cliquet option pricing}\label{sec3}

This section is dedicated to the pricing of cliquet options in the L\'{e}vy
jump-diffusion stock price model presented in Section~\ref{sec2}. In accordance to
(1.1) in \cite{3}, we consider a cliquet option with payoff
\begin{align*}
H_{T} =K+K \max \Biggl\{ g, \sum_{k=1}^{n}
\min \{ c, R_{k} \} \Biggr\}
\end{align*}
where $T$ is the maturity time, $K$ denotes the notional (the initial
investment), $g$ is the guaranteed rate at maturity, $c\geq0$ is the local
cap and $R_{k}$ is the return process defined in (\ref{eq2.10}). This option is of
monthly sum cap style with credited rate based on the sum of the
monthly-capped rates \cite{3}. Verify that the payoff $H_{T}$ is
globally-floored and locally-capped. A popular choice in the insurance
industry is to take $g=0$ (globally-floored by zero) and $n=12$
(monthly-capped by $c$). Further, note that the payoff $H_{T}$ actually is a
function of multiple random variables, i.e. $H_{T} =h  ( R_{1},\dots,
R_{n}  ) = \overline{h}  ( S_{t_{0}},\dots, S_{t_{n}}  )$
wherein $h$ and $\overline{h}$ are appropriately defined functions while
the resetting times of the cliquet option are ordered as follows $0< t_{0}
< t_{1} <\cdots< t_{n} \leq T$. In this regard, a notation like
$H_{t_{0},\dots, t_{n}}  ( T  )$ might be more intuitive than
simply writing $H_{T}$. However, by a case distinction we observe
\begin{align*}
H_{T} =K \max \Biggl\{ 1+g,1+ \sum_{k=1}^{n}
\min \{ c, R_{k} \} \Biggr\} =K \Biggl( 1+g+ \max \Biggl\{ 0, \sum
_{k=1}^{n} Z_{k} \Biggr\}
\Biggr)
\end{align*}
where $Z_{k} := \min  \{ c, R_{k}  \} - {g} / {n}$ denote
i.i.d. random variables. Moreover, we introduce a bank account $d B_{t} =r
B_{t} \,dt$ with constant interest rate $r>0$ and initial capital $B_{0} =1$,
i.e. $B_{t} = e^{rt}$. Then the price at time $t\leq T$ of a cliquet option
with payoff $H_{T}$ at maturity $T$ is the discounted risk-neutral
conditional expectation of the payoff,~i.e.
\begin{align*}
C_{t} = e^{-r  ( T-t  )} \mathbb{E}_{\mathbb{Q}} ( H_{T}
| \mathcal{F}_{t} ).
\end{align*}
Combining the latter equations, we obtain
%
\begin{align}
\label{eq3.1} C_{0} =K e^{-rT} \Biggl( 1+g+
\mathbb{E}_{\mathbb{Q}} \Biggl[ \max \Biggl\{ 0, \sum
_{k=1}^{n} Z_{k} \Biggr\} \Biggr] \Biggr)
\end{align}
which shows that the considered cliquet option with payoff $H_{T}$
essentially is a plain-vanilla call option with strike zero written on the
basket-style underlying $\sum_{k=1}^{n} Z_{k}$.

\begin{prop}[Cliquet option price]
\label{prop3.1}
 {Let} $k\in  \{
1,\dots,n  \}$ {and consider the independent and identically
distributed random variables} $Z_{k} = \min  \{ c, R_{k}  \} - {g}
/ {n}$ {where} $c\geq0$ {is the local cap}, $R_{k}$
{is the return process defined in} (\ref{eq2.10}) {and} $g$
{is the guaranteed rate at maturity. Denote the maturity time by}
$T$, {the notional by} $K$ {and the risk-less interest rate
by} $r$. {Then the price at time zero of a cliquet option with
payoff} $H_{T}$ {can be represented as}
%
\begin{align}
\label{eq3.2} C_{0} =K e^{-rT} \Biggl( 1+g+
\frac{n} {2} \mathbb{E}_{\mathbb{Q}} [ Z_{1} ] +
\frac{1} {\pi} \int_{0^{ +}}^{\infty} \frac{1-  \mathfrak{Re}  (
\phi_{Z}  ( x  )  )} {x^{2}}
\,dx \Biggr)
\end{align}
{where} $\mathfrak{Re}$ {denotes the real part and the characteristic
function} $\phi_{Z}  ( x  )$ {is defined via}
%
\begin{align}
\label{eq3.3} \phi_{Z} ( x ) := \prod_{k=1}^{n}
\phi_{Z_{k}} ( x ) = \prod_{k=1}^{n}
\mathbb{E}_{\mathbb{Q}} \bigl[ e^{ix Z_{k}} \bigr] = \bigl(
\phi_{Z_{1}} ( x ) \bigr)^{n} = \bigl( \mathbb{E}_{\mathbb{Q}}
\bigl[ e^{ix Z_{1}} \bigr] \bigr)^{n}.
\end{align}
\end{prop}

More explicit expressions for $\phi_{Z}  ( x  )$ and
$\mathbb{E}_{\mathbb{Q}}  [ Z_{1}  ]$ are derived in several
propositions below.

\begin{proof} The proof essentially follows the same lines as the proof
of Proposition 3.1 in~\cite{3} whereas our proof does not make use of the
Rademacher random variable introduced in \cite{3}. To begin with, we recall
that
\begin{align*}
&\max \{ 0,a \} = \frac{a+  \llvert  a  \rrvert } {2},
\\
& \llvert a \rrvert = \frac{2} {\pi} \int_{0^{ +}}^{\infty}
\frac{1-
\cos  ( ax  )} {x^{2}} \,dx = \frac{1} {\pi} \int_{0^{ +}}^{\infty}
\frac{2- e^{iax} - e^{-iax}} {x^{2}} \,dx
\end{align*}
similar to (3.2) and (3.3) in \cite{3}. As a consequence, we deduce
\begin{align*}
\mathbb{E}_{\mathbb{Q}} \Biggl[ \max \Biggl\{ 0, \sum
_{k=1}^{n} Z_{k} \Biggr\} \Biggr] = \sum
_{k=1}^{n} \frac{\mathbb{E}_{\mathbb{Q}}  [
Z_{k}  ]} {2} + \int
_{0^{ +}}^{\infty} \frac{2- \phi_{Z}  ( x
 ) - \phi_{Z}  ( -x  )} {2\pi x^{2}} \,dx
\end{align*}
where the characteristic function $\phi_{Z}  ( x  )$ is such as
defined in (\ref{eq3.3}). In the derivation of the latter equation, we used the
fact that $Z_{1},\dots, Z_{n}$ are i.i.d. random variables under
$\mathbb{Q}$. Since
\begin{align*}
\frac{1} {2} \bigl( \phi_{Z} ( x ) + \phi_{Z} ( -x
) \bigr) &= \mathbb{E}_{\mathbb{Q}} \Biggl[ \cos \Biggl( x \sum
_{k=1}^{n} Z_{k} \Biggr) \Biggr]
\\
&=\mathfrak{Re} \Biggl( \mathbb{E}_{\mathbb{Q}} \Biggl[ \exp \Biggl\{ ix \sum
_{k=1}^{n} Z_{k} \Biggr\}
\Biggr] \Biggr) =\mathfrak{Re} \bigl( \phi_{Z} ( x ) \bigr)
\end{align*}
we get
\begin{align*}
\mathbb{E}_{\mathbb{Q}} \Biggl[ \max \Biggl\{ 0, \sum
_{k=1}^{n} Z_{k} \Biggr\} \Biggr] = \sum
_{k=1}^{n} \frac{\mathbb{E}_{\mathbb{Q}}  [
Z_{k}  ]} {2} + \int
_{0^{ +}}^{\infty} \frac{1- \mathfrak{Re}  ( \phi_{Z}
 ( x  )  )} {\pi x^{2}} \,dx.
\end{align*}
Substituting this into (\ref{eq3.1}) leads us to (\ref{eq3.2}). \end{proof}

The remaining challenge now consists in finding appropriate computation
techniques for the entities $\mathbb{E}_{\mathbb{Q}}  [ Z_{1}  ]$
and $\phi_{Z}  ( x  )$ emerging in (\ref{eq3.2}). In the subsequent
sections, we present different methods to derive expressions for the
mentioned entities. Similar to the notation introduced in Proposition~\ref{prop2.9},
for arbitrary $k\in  \{ 1,\dots,n  \}$ we set $\tau= t_{k} -
t_{k-1}$ in the following. We also assume that the jump amplitudes are
normally distributed, as pointed out in Example~\ref{exmp2.2}.

\subsection{Cliquet option pricing with distribution
functions}\label{sec3.1}

Let us first apply a method involving probability distribution functions
(cf. \cite{3}). We initially investigate the treatment of $\phi_{Z}  ( x
 )$ as defined in (\ref{eq3.3}).

\begin{prop}\label{prop3.2}
{Suppose that} $Z_{k} = \min  \{ c,
R_{k}  \} - {g} / {n}$ {where} $k\in  \{ 1,\dots,n
 \}$. {Then the characteristic function of} $Z_{k}$
{under} $\mathbb{Q}$ {can be represented as}
%
\begin{align}
\label{eq3.4} \phi_{Z_{k}} ( x ) &= e^{-ix  ( 1+ {g} / {n}  )} \Biggl(
e^{ix  ( 1+c  )}
\nonumber\\
&\quad -ix e^{-\lambda\tau} \sum_{m=0}^{\infty}
\frac{ ( \lambda\tau  )^{m}} {m!} \int_{0^{ +}}^{1+c} e^{ixw}
\varPhi \biggl( \frac{\ln  ( w  ) -\gamma\tau-m\mu} {\sqrt{\sigma^{2}
\tau+m \delta^{2}}} \biggr) \,dw \Biggr)
\end{align}
{where} $\varPhi$ {denotes the standard normal cumulative
distribution function}.
\end{prop}

\begin{proof} By a case distinction, we find that the distribution
function of $Z_{k}$ is given by
%
\begin{align}
\label{eq3.5} \mathbb{Q} ( Z_{k} >\xi ) = \mathbb{Q} (
R_{k} - {g} / {n} >\xi )
\end{align}
if $R_{k} \leq c$ and $\xi\leq c- {g} / {n}$, whereas $\mathbb{Q}  (
Z_{k} >\xi  ) =0$ otherwise (cf. (3.15) in \cite{3}). Since $R_{k} >-1$
$\mathbb{Q}$-a.s. for all $k$, we deduce $Z_{k} >-1- {g} / {n}$
$\mathbb{Q}$-a.s. for all $k$. Thus, $Z_{k} +1+ {g} / {n} >0$
$\mathbb{Q}$-a.s. for all $k$. With respect to (\ref{eq3.5}), we obtain
%
\begin{align}
\label{eq3.5a} \mathbb{Q} ( Z_{k} +1+ {g} / {n} >w ) = \mathbb{Q} (
Z_{k} >w-1- {g} / {n} ) = \mathbb{Q} ( R_{k} >w-1 )
\tag{3.5a}
\end{align}
if $R_{k} \leq c$ and $w\leq1+c$, whereas $\mathbb{Q}  ( Z_{k} +1+ {g}
/ {n} >w  ) =0$ otherwise. Further on, verify that for the
characteristic function the following relation holds
%
\begin{align}
\label{eq3.6} \phi_{Z_{k} +1+ {g} / {n}} ( x ) e^{-ix  ( 1+ {g} / {n}
 )} = \phi_{Z_{k}}
( x ).
\end{align}
Moreover, we recall that for any random variable $\varLambda \geq0$ with
finite first moment, its characteristic function can be represented as
\begin{align*}
\phi_{\varLambda} ( x ) =1+ix \int_{0}^{\infty}
e^{ixu} \mathbb{ Q} ( \varLambda >u ) \,du.
\end{align*}
(This equality follows from integration by parts; cf. Eq. (3.14) in \cite{3}.)
Combining the latter equation with (\ref{eq3.6}) and (\ref{eq3.5a}), we deduce
%
\begin{align}
\label{eq3.7} \phi_{Z_{k}} ( x ) = e^{-ix  ( 1+ {g} / {n}  )} \Biggl( 1+ix \int
_{0}^{1+c} e^{ixw} \mathbb{ Q} (
R_{k} >w-1 ) \,dw \Biggr)
\end{align}
which can be rewritten as
\begin{align*}
\phi_{Z_{k}} ( x ) = e^{-ix  ( 1+ {g} / {n}  )} \Biggl( e^{ix  ( 1+c  )} -ix \int
_{0}^{1+c} e^{ixw} \mathbb{ Q} (
R_{k} \leq w-1 ) \,dw \Biggr).
\end{align*}
Merging (\ref{eq2.11}) into the latter equation while noting that in (\ref{eq2.11}) it
holds $\xi>-1$, we finally end up with (\ref{eq3.4}). \end{proof}

If we insert (\ref{eq3.4}) into (\ref{eq3.3}), we eventually get a representation for
the characteristic function $\phi_{Z}  ( x  )$. Let us proceed with the computation of $\mathbb{E}_{\mathbb{Q}}  [
Z_{k}  ]$.

\begin{prop}\label{prop3.3}
 {Suppose that} $Z_{k} = \min  \{ c,
R_{k}  \} - {g} / {n}$ {where} $k\in  \{ 1,\dots,n
 \}$. {Then the first moment of} $Z_{k}$ {under}
$\mathbb{Q}$ {is given by}
%
\begin{align}
\label{eq3.8} \mathbb{E}_{\mathbb{Q}} [ Z_{k} ] =c-
\frac{g} {n} - e^{-\lambda\tau} \sum_{m=0}^{\infty}
\frac{ ( \lambda\tau  )^{m}}{
m!} \int_{0^{ +}}^{1+c} \varPhi \biggl(
\frac{\ln  ( w  )
-\gamma\tau-m\mu} {\sqrt{\sigma^{2} \tau+m \delta^{2}}} \biggr) \,dw
\end{align}
{where} $\varPhi$ {denotes the standard normal distribution
function}.
\end{prop}

\begin{proof} In accordance to Proposition 2.4 in \cite{4}, we have
%
\begin{align}
\label{eq3.9} \mathbb{E}_{\mathbb{Q}} [ Z_{k} ] =
\frac{1} {i}  \frac{\partial} {\partial x} \bigl( \phi_{Z_{k}} ( x )
\bigr) \bigg\rrvert _{x=0}.
\end{align}
A substitution of (\ref{eq3.7}) into (\ref{eq3.9}) yields
\begin{align*}
\mathbb{E}_{\mathbb{Q}} [ Z_{k} ] =c- \frac{g} {n} - \int
_{0^{
+}}^{1+c} \mathbb{Q} ( R_{k} \leq w-1 )
\,dw.
\end{align*}
We ultimately put (\ref{eq2.11}) into the latter equation and receive (\ref{eq3.8}).
\end{proof}

\subsection{Cliquet option pricing with Fourier transform
techniques}\label{sec3.2}

There is an alternative method to derive expressions for
$\mathbb{E}_{\mathbb{Q}}  [ Z_{k}  ]$, $\phi_{Z}  ( x  )$
and $C_{0}$ involving Fourier transforms and the L\'{e}vy--Khinchin formula.
In the following, we present this method.

\begin{prop}\label{prop3.4} {Suppose that} $Z_{k} = \min  \{ c,
R_{k}  \} - {g} / {n}$ {where} $k\in  \{ 1,\dots,n
 \}$ {and let} $a>0$ {be a finite real-valued dampening\querymark{Q1}
parameter. Then the first moment of} $Z_{k}$ {under} $\mathbb{Q}$
{can be represented as}
%
\begin{align}
\label{eq3.10} \mathbb{E}_{\mathbb{Q}} [ Z_{k} ] =c-
\frac{g} {n} - \frac{1}{
2\pi} \int_{\mathbb{R}}
\frac{ ( c+1  )^{1+a+iy}} { ( a+iy
 )  ( 1+a+iy  )} \phi_{X_{\tau}} ( ia-y ) \,dy
\end{align}
{where the characteristic function} $\phi_{X_{\tau}}$ {is
given by}
%
\begin{align}
\label{eq3.11} \phi_{X_{\tau}} ( ia-y ) &= e^{-\lambda\tau} \sum
_{m=0}^{\infty} \frac{ ( \lambda\tau  )^{m}} {m!}
\nonumber
\\
&\quad \times \exp \biggl\{ ( a+iy ) \biggl( \frac{1} {2} \bigl(
\sigma^{2} \tau+m \delta^{2} \bigr) ( a+iy ) -\gamma\tau-m\mu
\biggr) \biggr\}.
\end{align}
\end{prop}

\begin{proof} First of all, verify that
%
\begin{align}
\label{eq3.12} \min \{ c, R_{k} \} =- \max \{ -c,- R_{k} \}
=c- \max \{ 0,c- R_{k} \} =c- [ c- R_{k} ]^{+}
\end{align}
which implies
\begin{align*}
\mathbb{E}_{\mathbb{Q}} [ Z_{k} ] =c- {g} / {n} -
\mathbb{E}_{\mathbb{Q}} \bigl[ ( c- R_{k} )^{+} \bigr].
\end{align*}
Hence, the evaluation of $\mathbb{E}_{\mathbb{Q}}  [ Z_{k}  ]$ is
equivalent to the evaluation of a put option with underlying $R_{k}$ and
strike $c\geq0$. Taking (\ref{eq2.10}) into account, we receive
\begin{align*}
\mathbb{E}_{\mathbb{Q}} [ Z_{k} ] =c- {g} / {n} -
\mathbb{E}_{\mathbb{Q}} \bigl[ \bigl( c+1- e^{X_{\tau}}
\bigr)^{+} \bigr]
\end{align*}
where $\tau= t_{k} - t_{k-1}$ and $X$ is the real-valued L\'{e}vy process
introduced in (\ref{eq2.3}). Furthermore, we define the function
\begin{align*}
\zeta ( u ) := e^{au} \bigl( c+1- e^{u}
\bigr)^{+}
\end{align*}
with a finite real-valued dampening\querymark{Q2} parameter $a>0$. Since $\zeta\in
\mathcal{L}^{1}  ( \mathbb{R}  )$, its Fourier transform exists
and reads as
\begin{align*}
\hat{\zeta} ( y ) = \frac{ ( c+1  )^{1+a+iy}} { (
a+iy  )  ( 1+a+iy  )}.
\end{align*}
Using the inverse Fourier transform along with Fubini's theorem, we get
\begin{align*}
\mathbb{E}_{\mathbb{Q}} \bigl[ \bigl( c+1- e^{X_{\tau}}
\bigr)^{+} \bigr] = \mathbb{E}_{\mathbb{Q}} \bigl[ e^{-a X_{\tau}}
\zeta ( X_{\tau} ) \bigr] = \frac{1} {2\pi} \int_{\mathbb{R}}
\hat{\zeta} ( y ) \mathbb{E}_{\mathbb{Q}} \bigl[ e^{-  ( a+iy  ) X_{\tau}} \bigr] \,dy
\end{align*}
which implies (\ref{eq3.10}). What remains is the computation of the characteristic
function $\phi_{X_{\tau}}$. It holds
\begin{align*}
\phi_{X_{\tau}} ( ia-y ) = \mathbb{E}_{\mathbb{Q}} \bigl[ e^{-
 ( a+iy  ) X_{\tau}}
\bigr] = \int_{\mathbb{R}} e^{-  ( a+iy
 ) x} f_{X_{\tau}} ( x ) \,dx,
\end{align*}
such that (\ref{eq2.13}) yields
\begin{align*}
\phi_{X_{\tau}} ( ia-y ) = e^{-\lambda\tau} \sum
_{m=0}^{\infty} \frac{ ( \lambda\tau  )^{m}} {m!} \int
_{\mathbb{R}} e^{-  (
a+iy  ) x} \varphi_{\gamma\tau+m\mu, \sigma^{2} \tau+m \delta^{2}} ( x ) \,dx.
\end{align*}
We finally perform the integration while noting that
%
\begin{align}
\label{eq3.13} \int_{\mathbb{R}} e^{bx}
\varphi_{\mu, \sigma^{2}} ( x ) \,dx = \exp \biggl\{ \mu b+ \frac{1} {2}
\sigma^{2} b^{2} \biggr\}
\end{align}
(with arbitrary $b \in \mathbb{C}$, $\mu \in \mathbb{R}$, $\sigma\in
\mathbb{R}^{+}$) and end up with (\ref{eq3.11}). \end{proof}

It is possible to derive an alternative representation for the
characteristic function $\phi_{X_{\tau}}$ by using (\ref{eq2.4}), (\ref{eq2.5}), Example~\ref{exmp2.2} and the equality $ ( ia-y  )^{2} =-  ( a+iy  )^{2}$.
If we do so, we obtain
\begin{align*}
\phi_{X_{\tau}} ( ia-y ) = \exp \biggl\{ \tau \!\biggl( - ( a+iy ) \gamma+
\frac{1} {2} \sigma^{2} ( a+iy )^{2} +\lambda
e^{ ( a+iy  )  [  ( a+iy  ) {\delta^{2}} /
{2} -\mu  ]} -\lambda \!\biggr) \!\biggr\}
\end{align*}
instead of (\ref{eq3.11}). In contrast to (\ref{eq3.11}), in the latter equation the series
expansion has vanished.

Our argumentation in the proof of Proposition~\ref{prop3.4} motivates the following
considerations.

\begin{prop}\label{prop3.5} {Suppose that} $Z_{k} = \min  \{ c,
R_{k}  \} - {g} / {n}$ {with} $k\in  \{ 1,\dots,n  \}$
{and} $c\geq0$. {Then the characteristic function of} $Z_{k}$
{under} $\mathbb{Q}$ {reads as}
%
\begin{align}
\label{eq3.14} \phi_{Z_{k}} ( x ) = e^{-ix {g} / {n}} \Biggl(
e^{ixc} + \int_{-\infty}^{\ln  ( 1+c  )} \bigl[
e^{ix  ( e^{u} -1
 )} - e^{ixc} \bigr] f_{X_{\tau}} ( u ) \,du \Biggr)
\end{align}
{where the density} $f_{X_{\tau}}$ of $X_{\tau}$ {under}
$\mathbb{Q}$ {is such as given in} (\ref{eq2.9}).
\end{prop}

\begin{proof} By the definition of the characteristic function (recall
(\ref{eq2.4})), we get
\begin{align*}
\phi_{Z_{k}} ( x ) = e^{-ix {g} / {n}} \mathbb{E}_{\mathbb{Q}} \bigl[
e^{ix \min  \{ c, R_{k}  \}} \bigr].
\end{align*}
Taking (\ref{eq3.12}) and (\ref{eq2.10}) into account, the latter equation can be expressed as
\begin{align*}
\phi_{Z_{k}} ( x ) = e^{-ix {g} / {n}} \mathbb{E}_{\mathbb{Q}} \bigl[
e^{ix  ( c-  [ 1+c- e^{X_{\tau}}  ]^{+}  )} \bigr] = e^{-ix {g} / {n}} \int_{\mathbb{R}}
e^{ix  ( c-  [ 1+c-
e^{u}  ]^{+}  )} f_{X_{\tau}} ( u ) \,du
\end{align*}
where the density $f_{X_{\tau}}$ is such as given in (\ref{eq2.9}). Next, verify
that
\begin{align*}
\bigl[ 1+c- e^{u} \bigr]^{+} = \bigl( 1+c-
e^{u} \bigr) \mathbh{1}_{u\leq
\ln  ( 1+c  )}
\end{align*}
(where $\mathbh{1}$ denotes the indicator function) which implies
\begin{align*}
\phi_{Z_{k}} ( x ) = e^{-ix {g} / {n}} \Biggl( \int_{-\infty}^{\ln  ( 1+c  )}
e^{ix  ( e^{u} -1  )} f_{X_{\tau}} ( u ) \,du + e^{ixc} \int
_{\ln  ( 1+c
 )}^{\infty} f_{X_{\tau}} ( u ) \,du \Biggr).
\end{align*}
Since the last integral can be rewritten as
\begin{align*}
\int_{\ln  ( 1+c  )}^{\infty} f_{X_{\tau}} ( u ) \,du =1- \int
_{-\infty}^{\ln  ( 1+c  )} f_{X_{\tau}} ( u ) \,du
\end{align*}
we eventually obtain (\ref{eq3.14}). \end{proof}

There is an alternative method involving (\ref{eq3.9}) to derive an expression for
$\mathbb{E}_{\mathbb{Q}}  [ Z_{k}  ]$ which is presented in the following.

\begin{cor}\label{cor3.6} In the setup of Proposition~\ref{prop3.5},
we receive the representation
%
\begin{align}
\label{eq3.15} \mathbb{E}_{\mathbb{Q}} [ Z_{k} ] &=c-
\frac{g} {n} + e^{-\lambda\tau} \sum_{m=0}^{\infty}
\frac{ ( \lambda\tau  )^{m}}{
m!} \biggl[ \exp \biggl\{ \biggl( \gamma+ \frac{\sigma^{2}} {2}
\biggr) \tau+ \biggl( \mu+ \frac{\delta^{2}} {2} \biggr) m \biggr\} \varPhi \bigl(
\kappa_{1}^{m} \bigr)
\nonumber
\\
&\quad - ( 1+c ) \varPhi \bigl( \kappa_{2}^{m} \bigr)
\biggr]
\end{align}
{wherein} $\varPhi$ {denotes the standard normal distribution
function and}
\begin{align*}
\kappa_{1}^{m} := \frac{\ln  ( 1+c  ) -\gamma\tau-m\mu-
\sigma^{2} \tau-m \delta^{2}} {\sqrt{\sigma^{2} \tau+m \delta^{2}}},\qquad
\kappa_{2}^{m} := \frac{\ln  ( 1+c  ) -\gamma\tau-m\mu}{
\sqrt{\sigma^{2} \tau+m \delta^{2}}}.
\end{align*}
\end{cor}

\begin{proof} A substitution of (\ref{eq3.14}) into (\ref{eq3.9}) yields
%
\begin{align}
\label{eq3.16} \mathbb{E}_{\mathbb{Q}} [ Z_{k} ] =c-
\frac{g} {n} + \int_{-\infty}^{\ln  ( 1+c  )} \bigl[
e^{u} -1-c \bigr] f_{X_{\tau}} ( u ) \,du.
\end{align}
Taking (\ref{eq2.13}) into account, we obtain the equalities
\begin{gather*}
\int_{-\infty}^{\ln  ( 1+c  )} e^{u} f_{X_{\tau}} (
u ) \,du = e^{ ( \gamma-\lambda+ \frac{\sigma^{2}} {2}  ) \tau} \sum_{m=0}^{\infty}
e^{ ( \mu+ {\frac{\delta^{2}} {2}}  ) m} \frac{ ( \lambda\tau  )^{m}} {m!} \varPhi \bigl( \kappa_{1}^{m}
\bigr),
\\
\int_{-\infty}^{\ln  ( 1+c  )} f_{X_{\tau}} ( u )
\,du = e^{-\lambda\tau} \sum_{m=0}^{\infty}
\frac{ ( \lambda\tau  )^{m}}{
m!} \varPhi \bigl( \kappa_{2}^{m} \bigr)
\end{gather*}
where the arguments $\kappa_{1}^{m}$ and $\kappa_{2}^{m}$ are such as
defined in the sequel of (\ref{eq3.15}). Putting the latter equations into (\ref{eq3.16}),
we ultimately get (\ref{eq3.15}). \end{proof}

Note that the expressions in (\ref{eq3.4}), (\ref{eq3.8}), (\ref{eq3.10}), (\ref{eq3.14}) and (\ref{eq3.15})
altogether are independent of $k$. This is not a surprising observation,
since $Z_{1},\dots, Z_{n}$ are i.i.d. random variables and we have chosen
equidistant resetting times with $\tau= t_{k} - t_{k-1}$.

Inspired by the Fourier transform techniques applied in the proof of
Proposition~\ref{prop3.4}, we now focus on the derivation of an alternative
representation for the cliquet option price $C_{0}$ given in Eq.~(\ref{eq3.1}).

\begin{thm}[Fourier transform cliquet option price]
\label{thm3.7} {Let}
$k\in  \{ 1,\dots,n  \}$ {and consider the independent and
identically distributed random variables} $Z_{k} \,{=}\, \min  \{ c, R_{k}
 \} - {g} / {n}$ {where} $c\geq0$ {is the local cap},
$g$ {is the guaranteed rate at maturity and} $R_{k}$ {is the
return process defined in} (\ref{eq2.10}). {For arbitrary} $n
\in \mathbb{N}$ {we set} $\varrho:=nc-g$ {and denote
the maturity time by} $T$, {the notional by} $K$ {and the
risk-less interest rate by} $r$. {Then the price at time zero of a
cliquet option paying}
\begin{align*}
H_{T} =K \Biggl( 1+g+ \max \Biggl\{ 0, \sum
_{k=1}^{n} Z_{k} \Biggr\} \Biggr)
\end{align*}
{at maturity can be represented as}
%
\begin{align}
\label{eq3.17} C_{0} &=K e^{-rT} \Biggl[ 1+g
\nonumber
\\
&\quad + \int_{0^{ +}}^{\infty} \frac{1+iy\varrho-
e^{iy\varrho}} {2\pi y^{2}} \Biggl(
1+ \int_{-\infty}^{\ln  ( 1+c
 )} \bigl[ e^{iy  ( e^{u} -1-c  )} -1 \bigr]
f_{X_{\tau}} ( u ) \,du \Biggr)^{n} \,dy \Biggr]
\end{align}
{where} $f_{X_{\tau}}  ( u  )$ {constitutes the
probability density function claimed in} (\ref{eq2.9}).
\end{thm}

\begin{proof} Suppose that the cliquet option price $C_{0}$ is such as
given in (\ref{eq3.1}). We only need to evaluate the expectation
\begin{align*}
J:= \mathbb{E}_{\mathbb{Q}} \Biggl[ \Biggl( \sum_{k=1}^{n}
Z_{k} \Biggr)^{+} \Biggr]
\end{align*}
appearing in (\ref{eq3.1}). Taking the definition of $Z_{k}$ and (\ref{eq3.12}) into
account, we obtain
\begin{align*}
J= \mathbb{E}_{\mathbb{Q}} \Biggl[ \Biggl( \varrho- \sum
_{k=1}^{n} [ c- R_{k} ]^{+}
\Biggr)^{+} \Biggr]
\end{align*}
where $\varrho=nc-g$ is a constant. Note that in the latter equation we
observe a basket-style composition of put options. Let us further introduce
the function $\vartheta  ( x  ) :=  ( \varrho-x
 )^{+} \in \mathcal{L}^{1}  ( \mathbb{R}^{+}  )$ as well as
the non-negative random variable
\begin{align*}
D:= \sum_{k=1}^{n} [ c- R_{k}
]^{+}
\end{align*}
such that we may write
\begin{align*}
J= \mathbb{E}_{\mathbb{Q}} \bigl[ \vartheta ( D ) \bigr].
\end{align*}
An application of the inverse Fourier transform yields
\begin{align*}
J= \frac{1} {2\pi} \int_{\mathbb{R}^{+}} \hat{\vartheta} ( y )
\mathbb{E}_{\mathbb{Q}} \bigl[ e^{-iyD} \bigr] \,dy
\end{align*}
where
\begin{align*}
\hat{\vartheta} ( y ) = \frac{1+iy\varrho- e^{iy\varrho}}{
y^{2}}
\end{align*}
constitutes the Fourier transform of $\vartheta$. In the next step, we
compute the characteristic function of $D$. Taking the definition of $D$
and the $\mathbb{Q}$-independence of the random variables $R_{1},\dots,
R_{n}$ into account, we deduce
\begin{align*}
\mathbb{E}_{\mathbb{Q}} \bigl[ e^{-iyD} \bigr] = \prod
_{k=1}^{n} \mathbb{E}_{\mathbb{Q}} \bigl[
e^{-iy  [ c- R_{k}  ]^{+}} \bigr].
\end{align*}
With respect to (\ref{eq2.10}), we obtain
\begin{align*}
\mathbb{E}_{\mathbb{Q}} \bigl[ e^{-iyD} \bigr] = \prod
_{k=1}^{n} \mathbb{E}_{\mathbb{Q}} \bigl[
e^{-iy  [ 1+c- e^{X_{\tau}}  ]^{+}} \bigr] = \prod_{k=1}^{n}
\int_{\mathbb{R}} e^{-iy  [ 1+c- e^{u}
 ]^{+}} f_{X_{\tau}} ( u ) \,du
\end{align*}
where $\tau= t_{k} - t_{k-1}$ and $f_{X_{\tau}}  ( u  )$ is such
as given in (\ref{eq2.9}). By a case distinction, we get
\begin{align*}
\mathbb{E}_{\mathbb{Q}} \bigl[ e^{-iyD} \bigr] &= \prod
_{k=1}^{n} \Biggl( \int_{-\infty}^{\ln  ( 1+c  )}
e^{-iy  ( 1+c- e^{u}  )} f_{X_{\tau}} ( u ) \,du + \int_{\ln  ( 1+c  )}^{\infty}
f_{X_{\tau}} ( u ) \,du \Biggr)
\\
& = \prod_{k=1}^{n} \Biggl( 1+ \int
_{-\infty}^{\ln  ( 1+c  )} \bigl[ e^{iy  ( e^{u} -1-c
 )} -1 \bigr]
f_{X_{\tau}} ( u ) \,du \Biggr).
\end{align*}

Verify that the emerging integrand $ [ e^{iy  ( e^{u} -1-c  )}
-1  ] f_{X_{\tau}}  ( u  )$ is finite for all $u\in  [
-\infty, \ln  ( 1+c  )  ]$. Also note that the appearing
factors altogether are independent of~$k$. Combining the latter equations
with (\ref{eq3.1}), we finally get the asserted cliquet option price formula
(\ref{eq3.17}). \end{proof}

We recall that Fourier transform techniques have also been applied in the
context of cliquet option pricing in e.g. \cite{9} and \cite{11}.

\section{Hedging and Greeks}\label{sec4}

In this section, we are concerned with sensitivity analysis and the
computation of Greeks in our cliquet option pricing context. Let us start
with an investigation of the Greek \textit{Rho} which is defined as the
derivative of the cliquet option price $C_{0}$ with respect to the interest
rate $r$. Due to (\ref{eq3.2}), respectively (\ref{eq3.17}), we find
\begin{align*}
\rho:= \frac{\partial C_{0}} {\partial r} =-T C_{0}
\end{align*}
where $T$ denotes the maturity time of the option. Further on, both the
\textit{Delta} and \textit{Gamma} of the cliquet option vanish, i.e.
\begin{align*}
\Delta := \frac{\partial C_{0}} {\partial S_{0}} =0,\qquad
 \varGamma := \frac{\partial^{2} C_{0}} {\partial S_{0}^{2}} =0
\end{align*}
since we assumed $t_{0} \neq0$ in (\ref{eq2.10}) such that neither $R_{1}$ nor
$Z_{1}$ contains $S_{0}$. In accordance to Section 3.2 in \cite{3}, we claim
that in any cliquet option pricing context the most important Greek to
study is the \textit{Vega} which is defined as
\begin{align*}
V:= \frac{\partial C_{0}} {\partial\sigma}
\end{align*}
where $\sigma>0$ denotes the volatility parameter of the L\'{e}vy process
$X$ defined in (\ref{eq2.3}). In the Fourier transform framework studied in Section~\ref{sec3.2}, we get the following result.

\begin{prop}[Vega; Fourier transform case]\label{prop4.1} {Presume
that the density function} $f_{X_{\tau}}$ {is such as given in}
(\ref{eq2.13}) {and denote the density function of the normal
distribution by} $\varphi$. In the setup of Theorem~\ref{thm3.7},
we then find the following expression for the Vega of the cliquet option
%
\begin{align}
\label{eq4.1} V= \frac{n K e^{-rT}} {2\pi} \int_{0^{ +}}^{\infty}
\frac{1+iy  ( nc-g
 ) - e^{iy  ( nc-g  )}} {y^{2}} F_{y} ( \sigma )^{n-1} F_{y}
' ( \sigma ) \,dy
\end{align}
{where}
\begin{gather*}
F_{y} ( \sigma ) :=1+ \int_{-\infty}^{\ln  ( 1+c
 )}
\bigl[ e^{iy  ( e^{u} -1-c  )} -1 \bigr] f_{X_{\tau}} ( u ) \,du,
\\
F_{y} ' ( \sigma ) =\sigma \tau
e^{-\lambda\tau} \sum_{m=0}^{\infty}
\frac{ ( \lambda\tau  )^{m}} {m!  (
\sigma^{2} \tau+m \delta^{2}  )^{2}} \int_{-\infty}^{\ln  ( 1+c
 )}
\varphi_{\gamma\tau+m\mu, \sigma^{2} \tau+m \delta^{2}} ( u ) G_{y} ( u ) \,du,
\\
G_{y} ( u ) := \bigl[ e^{iy  ( e^{u} -1-c  )} -1 \bigr] \bigl[ (
u-\gamma\tau-m\mu )^{2} - \sigma^{2} \tau-m
\delta^{2} \bigr].
\end{gather*}
\end{prop}

\begin{proof} First of all, note that the only ingredient in (\ref{eq3.17}) which
contains the parameter $\sigma$ is the density function $f_{X_{\tau}}$.
Thus, from (\ref{eq3.17}) we deduce
\begin{align*}
V= \frac{n K e^{-rT}} {2\pi} \int_{0^{ +}}^{\infty}
\frac{1+iy\varrho-
e^{iy\varrho}} {y^{2}} F_{y} ( \sigma )^{n-1} F_{y}
' ( \sigma ) \,dy
\end{align*}
where $F_{y}  ( \sigma  )$ is as defined in the proposition and
the derivative $F_{y} '  ( \sigma  ) := {\partial F_{y}
 ( \sigma  )} / {\partial\sigma}$ reads as
\begin{align*}
F_{y} ' ( \sigma ) = \int_{-\infty}^{\ln  ( 1+c  )}
\bigl[ e^{iy  ( e^{u} -1-c  )} -1 \bigr] \frac{\partial
f_{X_{\tau}}  ( u  )} {\partial\sigma} \,du.
\end{align*}
Taking (\ref{eq2.13}) into account, we find
\begin{align*}
\frac{\partial f_{X_{\tau}}  ( u  )} {\partial\sigma} =\sigma \tau e^{-\lambda\tau} \sum
_{m=0}^{\infty} \frac{ ( \lambda\tau  )^{m}}{
m!} \frac{ ( u-\gamma\tau-m\mu  )^{2} - \sigma^{2} \tau-m
\delta^{2}} { ( \sigma^{2} \tau+m \delta^{2}  )^{2}}
\varphi_{\gamma\tau+m\mu, \sigma^{2} \tau+m \delta^{2}} ( u ).
\end{align*}
Putting the latter equations together, we obtain (\ref{eq4.1}) which completes the
proof. \end{proof}

In the distribution function context studied in Section~\ref{sec3.1}, we find the
following expression for the Vega.

\begin{prop}[Vega; distribution function case]
\label{prop4.2}
{Let us denote by} $\varphi_{0,1} = \varPhi'$ the probability density function
of the standard normal distribution. Then, in the setup of
Proposition~\ref{prop3.1}, we get the following representation for the Vega
%
\begin{align}
\label{eq4.2} V&=n \tau \sigma K e^{-rT-\lambda\tau} \sum
_{m=0}^{\infty} \frac{ (
\lambda\tau  )^{m}} {m!}
\nonumber
\\
&\quad \times \int_{0^{ +}}^{1+c} \varPsi ( w ) \Biggl[
\frac{1} {2} - \frac{1} {\pi} \int_{0^{ +}}^{\infty}
\frac{\mathfrak{Re}  (
i e^{ix  [ w-1- {g} / {n}  ]} \phi_{Z_{1}}  ( x  )^{n-1}
 )} {x} \,dx \Biggr] \,dw
\end{align}
{where the characteristic function} $\phi_{Z_{1}}  ( x  )$
{is such as given in} (\ref{eq3.4}) {and}
\begin{align*}
\varPsi ( w ) := \varphi_{0,1} \biggl( \frac{\ln  ( w
 ) -\gamma\tau-m\mu} {\sqrt{\sigma^{2} \tau+m \delta^{2}}} \biggr)
\frac{\ln  ( w  ) -\gamma\tau-m\mu} {\sqrt{ ( \sigma^{2}
\tau+m \delta^{2}  )^{3}}}.
\end{align*}
\end{prop}

\begin{proof} Taking (\ref{eq3.2}) into account, we get
\begin{align*}
V=K e^{-rT} \Biggl( \frac{n} {2} \frac{\partial \mathbb{E}_{\mathbb{Q}}
 [ Z_{1}  ]} {\partial\sigma} -
\frac{1} {\pi} \int_{0^{
+}}^{\infty} \mathfrak{Re}
\biggl( \frac{\partial \phi_{Z}  ( x  )}{
\partial\sigma} \biggr) x^{-2} \,dx \Biggr).
\end{align*}
Using (\ref{eq3.8}) and the ordinary chain rule, we obtain
\begin{align*}
\frac{\partial \mathbb{E}_{\mathbb{Q}}  [ Z_{1}  ]}{
\partial\sigma} &=\tau \sigma e^{-\lambda\tau} \sum
_{m=0}^{\infty} \frac{ ( \lambda\tau  )^{m}} {m!}
\\
&\quad \times \int
_{0^{ +}}^{1+c} \varphi_{0,1} \biggl(
\frac{\ln  ( w  ) -\gamma\tau-m\mu}{
\sqrt{\sigma^{2} \tau+m \delta^{2}}} \biggr) \frac{\ln  ( w  )
-\gamma\tau-m\mu} {\sqrt{ ( \sigma^{2} \tau+m \delta^{2}  )^{3}}} \,dw
\end{align*}
where $\varphi_{0,1} = \varPhi'$ denotes the probability density function of
the standard normal distribution. On the other hand, by (\ref{eq3.3}) and (\ref{eq3.4}) we
deduce
\begin{align*}
\frac{\partial} {\partial\sigma} \phi_{Z} ( x ) &=n \phi_{Z_{1}} ( x
)^{n-1} \frac{\partial} {\partial\sigma} \phi_{Z_{1}} ( x )
\\
& = %
\tau \sigma n i x \phi_{Z_{1}} ( x )^{n-1}
e^{-ix  ( 1+ {g}
/ {n}  ) -\lambda\tau} \sum_{m=0}^{\infty}
\frac{ ( \lambda\tau
 )^{m}} {m!}
\\
&\quad \times \int_{0^{ +}}^{1+c} e^{ixw}
\varphi_{0,1} \biggl( \frac{\ln  ( w  ) -\gamma\tau-m\mu} {\sqrt{\sigma^{2} \tau+m
\delta^{2}}} \biggr) \frac{\ln  ( w  ) -\gamma\tau-m\mu}{
\sqrt{ ( \sigma^{2} \tau+m \delta^{2}  )^{3}}} \,dw.
\end{align*}
Putting the latter equations together, we end up with the asserted
representation~(\ref{eq4.2}). \end{proof}

\section{Conclusions}\label{sec5}

In this paper, we investigated the pricing of a monthly sum cap style
cliquet option with underlying stock price modeled by a jump-diffusion
L\'{e}vy process with compound Poisson jumps. In Section~\ref{sec2}, we derived
representations for the probability density and distribution function of
the involved L\'{e}vy process. Moreover, we obtained semi-analytic
expressions for the cliquet option price by using the probability
distribution function of the driving L\'{e}vy process in Section~\ref{sec3.1} and by
an application of Fourier transform techniques in Section~\ref{sec3.2}. With view on
existing literature, the main contribution of the paper consists of the
Fourier transform cliquet option price formula provided in Theorem~\ref{thm3.7}. In
Section~\ref{sec4}, we concentrated on sensitivity analysis and computed the Greeks
Rho, Delta, Gamma and Vega.

A future research topic might consist in a transformation of the presented
techniques and results to a time-inhomogeneous L\'{e}vy process setup
which, in particular, contains a time (and state) dependent volatility
coefficient $\sigma  ( t  )$, respectively $\sigma  ( t, X_{t}
 )$, in order to obtain more realistic (implied) volatility surfaces.
In this context, we refer to Section 4 in \cite{16} as well as Sections 1.2.1
and 11.1.2 in \cite{4}.

To read more on cliquet option pricing in a pure-jump Meixner--L\'{e}vy
process setup, the reader is referred to the accompanying article \cite{8}.





\end{document}